\documentclass[12pt]{l4dc2023}

\title[ Data-driven Stochastic Output-Feedback Predictive Control]{Data-driven Stochastic Output-Feedback Predictive Control: Recursive Feasibility through Interpolated Initial Conditions}
\usepackage{times}

 \author{\Name{Guanru Pan} \Email{guanru.pan@tu-dortmund.de}\\
  \Name{Ruchuan Ou} \Email{ruchuan.ou@tu-dortmund.de}\\
  \Name{Timm Faulwasser} \Email{timm.faulwasser@ieee.org}\\
  \addr Institute of Energy Systems, Energy Efficiency and Energy
  Economics,\\ TU Dortmund University, Dortmund, 44227 Germany }

\usepackage{amsfonts}
\usepackage{amsmath}
\usepackage{mathtools}
\usepackage{mathrsfs}
\usepackage{soul,color}

\newcommand\scalemath[2]{\scalebox{#1}{\mbox{\ensuremath{\displaystyle #2}}}}

\usepackage{algorithm}
\usepackage{algorithmicx}  
\usepackage{algpseudocode}
\usepackage{booktabs}

\DeclareMathOperator*{\trace}{trace}
 
\newcommand{\mbb}[1]{\mathbb{#1 }}
\newcommand{\mbf}[1]{\mathbf{#1}} 
\newcommand{\mcl}[1]{\mathcal{#1}}

\newcommand{\pce}[1]{\mathsf{#1}}
\newcommand{\pcecoe}[2]{\mathsf{#1}^{#2}}
\newcommand{\relx}{(\omega)}

\newcommand{\trar}[2]{\mathbf{#1}_{[0,{#2}]}}

\newcommand{\inst}[1]{_{#1}}

\newcommand{\splx}[1]{\mcl{L}^2(\Omega, \mathcal{F}, \mu; \mathbb{R}^{#1})}
\newcommand{\spl}{\mcl{L}^2(\Omega, \mathcal{F}, \mu; \mathbb{R})} 
 
\newcommand{\diff}{\mathop{}\!\mathrm{d}}

\newcommand{\mean}{\mbb{E}}
\newcommand{\var}{\mbb{V}}

\newcommand{\hankel}{\mcl{H}}
 \newcommand{\covar}{\Sigma}

\newcommand{\Zf}{\mbb{Z}_{\text{f}}}

\newcommand{\Tini}{{T_{\text{ini}}}}
\newcommand{\lag}{\Tini}

\newcommand{\ini}{{\text{ini}}} 

\newcommand{\dimx}{{n_x}}
\newcommand{\dimy}{{n_y}}
\newcommand{\dimu}{{n_u}}
\newcommand{\dimw}{{n_w}}
\newcommand{\dimz}{{n_z}}
\newcommand{\dimv}{{n_v}}

\newcommand{\I}{\mathbb{I}}
\newcommand{\N}{\mathbb{N}}
\newcommand{\R}{\mathbb{R}}

\newtheorem{assumption}{Assumption}
\newcommand{\norm}[1]{\|{#1}\|_{\mcl{L}^2}}

\begin{document}

\maketitle

\begin{abstract}%
The paper investigates data-driven output-feedback  predictive control of linear systems subject to stochastic disturbances. The scheme relies on the recursive solution of a suitable data-driven reformulation of a stochastic Optimal Control Problem (OCP), which allows for  forward prediction and optimization of statistical distributions of inputs and outputs. Our approach avoids the use of parametric system models. Instead it is based on previously recorded data using a recently proposed  stochastic variant of Willems' fundamental lemma. The stochastic variant of the lemma is applicable to a large class of linear dynamics subject to stochastic disturbances of Gaussian and non-Gaussian nature.  To ensure recursive feasibility, the initial condition of the OCP---which consists of information about past inputs and outputs---is considered as an extra decision variable of the OCP. We provide sufficient conditions for recursive feasibility and closed-loop practical stability of the proposed scheme as well as performance bounds. Finally, a numerical example illustrates the efficacy and closed-loop properties of the proposed scheme. 
\end{abstract}


\section{Introduction}
Data-driven control based on the so-called fundamental lemma by \cite{Willems2005} has attracted a lot of research interest, see \cite{Markovsky21r,DePersis19} for recent reviews. The pivotal insight of the fundamental lemma is that any controllable LTI system can be characterized by its recorded input-output trajectories provided a persistency of excitation condition to be satisfied. In the context of predictive control this implies that the prediction of future input and output trajectories can be done based on measurements of past trajectories thus alleviating the need for system identification and state estimator  design \citep{Yang15,Coulson2019,L4DC-lian21a,L4DC-allibhoy20a}. Hence data-driven predictive control schemes are consider for different applications, e.g.~\citep{Carlet2022,Bilgic22,Wang2022a}. Rigorous stability  guarantees are provided by \cite{Berberich20,Berberich2021t,Bongard2022} for data-driven predictive control schemes with respect to deterministic LTI systems subjected to noisy measurements.

However, besides measurement noise, so far little has been done on data-driven predictive control for LTI systems subject to stochastic disturbances. One key challenge is the prediction of the future evolution of statistical distributions of the inputs (respectively input policies) and outputs in a data-driven fashion. One approach is to predict the expected value of the future trajectory by the fundamental lemma while handling the stochastic uncertainties offline by probabilistic constraint tightening \citep{Kerz21d} or directly omitting the stochastic uncertainties in the prediction \citep{Wang2022}. Alternatively, leveraging the framework of Polynomial Chaos Expansions (PCE) \citep{sullivan15introduction} a stochastic variant of Willems' fundamental lemma has been proposed \citep{Pan21s}. It allows to predict future statistical distributions of inputs and outputs via previously recorded data and knowledge about the distribution of the disturbance.

Extending the conceptual ideas of \cite{Pan21s} towards stochastic data-driven predictive control with guarantees, \cite{tudo:pan22a} present first results in a state-feedback setting, while the output-feedback case is discussed by \cite{Pan2022}. Therein, sufficient conditions of recursive feasibility and stability are provided by a data-driven design of terminal ingredients and a selection strategy of the initial condition which is similar to the model-based approach of \cite{farina13probabilistic,Farina2015}. 
The selection of initial conditions considers a binary choices: the current measured value or its predicted value based on the last optimal solution.

In this paper, we extend the data-driven stochastic  output-feedback  predictive scheme proposed by \cite{Pan2022} with an improved initialization strategy. Instead of the binary selection strategy described above, the proposed scheme interpolates between the measured value and the latest prediction in a continuous fashion. In model-based stochastic predictive control this has been considered by \cite{Koehler2022,Schlueter2022}. The main contribution of the present paper are sufficient conditions for stability and recursive feasibility of the proposed data-driven stochastic  output-feedback  scheme.

\section{Problem statement and preliminaries}
\label{sec:preliminary}
We investigate a data-driven stochastic output feedback approach to control of LTI systems with unknown system matrices.  Hence, we first detail the considered setting, then we recall the representation of $\mathcal L^2$ random variables via polynomial chaos, before we arrive at the stochastic fundamental lemma. 

\subsection{Considered system class}
We consider stochastic LTI systems in  AutoRegressive with eXtra input (ARX) form
\begin{equation}\label{eq:ARX}
	Y_{k} = \Phi Z_{k} + D U_{k} +W_{k},\quad Z_0 = Z_{\ini},
\end{equation}
with input $U_k \in \splx{\dimu}$, output $Y_k\in $ $\splx{\dimy}$,  disturbance~$W_k \in$ $\splx{\dimw}$ $(\dimw = \dimy)$, and extended state
\[Z\inst{k} \doteq \begin{bmatrix}
	\mbf{U}_{[k-\lag,k-1]}\\
	\mbf{Y}_{[k-\lag,k-1]} \end{bmatrix} = \left[ U_{k-\lag}^\top,U_{k-\lag +1}^\top , \cdots, U_{k-1}^\top,Y_{k-\lag}^\top, \cdots, Y_{k-1}^\top  \right]^\top\in \splx{\dimz},\]
 with $\dimz =\lag(\dimu+\dimy)$ which contains last $\lag$ inputs and outputs. 
 
Here, $\Omega$ denotes the sample space, $\mcl{F}$ is a $\sigma$ algebra, and $\mu$ is the considered probability measure. With the specification of the underlying probability space $(\Omega,\mcl{F},\mu)$ to be $\mcl{L}^2$, we restrict the consideration to  random variables with finite expectation and finite (co)-variance. Throughout this paper, the system matrices  $\Phi \in \R^{\dimy \times \dimz}$ and $D \in \R^{\dimy\times \dimu}$ are considered to be unknown, while the  statistical distributions of initial condition $Z_{\ini}$ and disturbances $\{W_{k}\}_{k \in \N}$ are supposed to be known. Furthermore, we consider that all $\{W_{k}\}_{k \in \N}$ are identical independently distributed ($i.i.d.$) with zero mean and finite co-variance, i.e., we assume for all $ k \in \N$,  $ \mean[W_k] =0$ and $ \covar[W_{k}]= \Sigma_W$.
 We remark that the initial condition and the disturbances are not restricted to be Gaussian.

 For a specific uncertainty outcome $\omega \in \Omega$, we denote the realization of $W_{k}$ as $w_{k} \doteq W_{k}(\omega)$. Likewise, the input, output, and extended state realizations are written as  $u_k \doteq U_k\relx$, $y_k \doteq Y_k\relx$, and $z_k \doteq Z_k\relx$, respectively. Moreover,  given $z_{\ini}$ and $w_k$, $k\in \N $, the stochastic system \eqref{eq:ARX}  induces the \textit{realization dynamics}  
\begin{equation}\label{eq:ARX_realization}
	y_{k} = \Phi z_{k} + D u_{k} +w_{k},\quad z_0 = z_{\ini}.
\end{equation}	
Throughout the paper,  we assume the input, output, and disturbance realizations $u_k$, $y_k$, and $w_{k-1}$ to be measured at time instant $k$. For the case of unmeasured disturbances, we refer to the disturbance estimation schemes tailored to ARX models, see~\citep{Pan21s,Wang2022} .

\begin{assumption}[Minimal state-space representation]\label{ass:minimal_state_system}~\\
	There  exists	 a minimal state-space representation 
	\begin{subequations}\label{eq:RVdynamics}
		\begin{align}
			X\inst{k+1} &= AX\inst{k} +BU\inst{k}+ EW\inst{k},\quad X_0=X_{\ini}  \\
			Y\inst{k} &= CX\inst{k} + DU\inst{k} + W\inst{k}, \label{eq:RVdynamics_y}
		\end{align}
	\end{subequations}
	with  $(A,B)$ controllable  and  $(A,C)$ observable such that for some initial condition $X_{\ini}$ and $W_k$, $k \in \N$, the input-output trajectories of  \eqref{eq:ARX} and  \eqref{eq:RVdynamics} coincide.
\end{assumption}
For insights into the problem of mapping  a ARX model to its minimal state-space representations, see \cite{Sadamoto2022,Wu2022}.
\subsection{Representation of random variables via polynomial chaos expansions}
It is well-known that for stochastic LTI systems subject to Gaussian disturbances, the evolution of statistical distributions of inputs (generated via affine policies) and outputs can be exactly represented by the first two moments, cf. \cite{farina13probabilistic,Farina2015}. However, this is not necessarily the case for non-Gaussian disturbances. Moreover, observe that already in a Gaussian setting moments constitute non-linear representations of random variables as any scalar Gaussian is given by the sum of its mean with its standard deviation ($=$ square root of variance) times a standard normal-distributed random variable. Alternatively, one could rely on scenario-based approach and   sampling strategies~\citep{Kantas2009,Tempo2013,Schildbach2014}. However, this induces substantial computational effort. 

Alternatively, we employ Polynomial Chaos Expansions (PCE) to provide a tractable linear surrogate of  \eqref{eq:ARX} by representing $\mathcal{L}^2$ random variables in a suitable polynomial basis.  PCE dates back to \cite{wiener38homogeneous}, and we refer to \cite{sullivan15introduction} for a general introduction, see, e.g.,  \citep{paulson14fast,Mesbah14,Ou21} for applications in control.

Consider an orthogonal polynomial basis $\{\phi^j\}_{j=0}^\infty$ which spans $\spl$, i.e.,
\[
\langle \phi^i, \phi^j \rangle \doteq \int_{\Omega} \phi^i(\omega)\phi^j(\omega) \diff \mu(\omega) = \delta^{ij}\norm{\phi^j}^2,
\]
where $\delta^{ij}$ is the Kronecker delta and $\norm{V}\doteq \sqrt{ \langle V,V\rangle}$ represents the induced form of $V \in \spl$.
With respect to the basis $\{\phi^j\}_{j=0}^\infty$,  a real-valued random scalar variable  $V\in \spl$ can be expressed as
$
 	V = \sum_{j=0}^{\infty}\pcecoe{v}{j} \phi^j$ with $\pcecoe{v}{j} = \langle V, \phi^j \rangle/\norm{\phi^j}^2,
$
 where $\pcecoe{v}{j} \in \R$ is called the $j$-th PCE coefficient. 
For a vector-valued random variable $V\in\splx{n_v}$ applying PCE component-wise the $j$-th coefficient is given as
$\pce{v}^j =$ $ \begin{bmatrix} \pcecoe{v}{1,j} & \pcecoe{v}{2,j} & \cdots & \pcecoe{v}{n_v,j} \end{bmatrix}^\top$ where  $\pcecoe{v}{i,j}$ is the $j$-th PCE coefficient of the component $V^i$ of $V$.

In practice one often terminates the PCE series after a finite number of terms for more efficient computation. However, this may lead to truncation errors \citep{muehlpfordt18comments}. Fortunately, random variables that follow some widely-used distributions admit exact finite-dimensional PCEs with only two terms in suitable polynomial bases, see \citep{koekoek96askey, xiu02wiener}. For example, the Legendre
basis is preferably chosen for uniformly-distributed random variables and Hermite polynomials are used for Gaussians. 

\begin{definition}[Exact PCE representation \citep{muehlpfordt18comments}]
The PCE of a  random variable $V\in \splx{n_v}$ is said to be exact  with dimension $L$ if
$	V - \sum_{j=0}^{L-1} \pcecoe{v}{j}\phi^j=0$. 
\end{definition}
 Notice that with an exact PCE of finite dimension, the expected value, variance and covariance of $V\in\splx{n_v}$ can be efficiently calculated from its PCE coefficients
\begin{equation}\label{eq:PCEmoments}
 \mean\big[V\big] = \pce{v}^0,\quad \var \big[V\big] = \sum_{j=1}^{L-1} (\pce{v}^{j})^2\norm{\phi^j}^2,
	\quad  \covar \big[V\big] = \sum_{j=1}^{L-1} \pce{v}^j\pce{v}^{j\top}\norm{\phi^j}^2,
\end{equation}
where $(\pce{v}^{j})^2 \doteq \pce{v}^{j} \circ \pce{v}^{j}$ denotes the Hadamard product \citep{Lefebvre20}. Finally, observe that (finite or infinite) PCEs of $\mathcal L^2$ random variables constitute linear representations of random variables. 

 To the end of reformulation of \eqref{eq:ARX} with finite dimensional PCEs,  we assume that $Z_{\ini}$ and $W_{k}, k\in \N^+$ admit exact PCEs in the basis $\{\phi^j_{\ini} \}_{j=0}^{L_\ini -1}$, respectively,  in the basis $\{\varphi_k^j\}_{j=0}^{L_w -1}$. 
Consider system~\eqref{eq:ARX} for a finite horizon $N \in \N$ and the finite-dimensional joint basis 
 	\begin{align}	 
	\{\phi^j\}_{j=0}^{L-1} = \left\{1, \{\phi^j_{\ini} \}_{j=1}^{L_\text{ini}-1}, \bigcup_{i=0}^{N-1} \{\varphi_{i}^j\}_{j=1}^{L_w-1} \right\}, \qquad L &= L_{\text{ini}} + N(L_w-1) \in \N^+.\label{eq:finite_basis}
\end{align}
Then, for all $k\in\I_{[0,N-1]}$, $U_k$, $Y_k$, and thus $Z_k$ also admit exact PCEs in the chosen basis if $U_k$ is designed as an affine feedback of $Z_k$, cf. the formal proofs given in \citep{Pan21s, Pan2022}. Note that the key aspect of $\{\phi^j\}_{j=0}^{L-1}$ is that it is the union of the bases for $Z_{\ini}$ and $W_{k}$, $ k \in \I_{[0,N-1]}$. Thus, PCE enables  uncertainty propagation over any finite prediction horizon. 
	
 Replacing all random variables in \eqref{eq:ARX} with their PCE representation in the basis $ \{\phi^j\}_{j=0}^{L-1}$ and projecting onto the basis functions $\phi^j$, we obtain the dynamics of the PCE coefficients.  For given $\pce{z}^j_{\ini}$ and $\pce{w}^j_{k}$, $k\in \N$, the \textit{dynamics of the PCE coefficients} read
\begin{equation}\label{eq:ARX_PCE}
	\pce{y}^j_{k} = \Phi \pce{z}^j_{k}+ D \pce{u}^j_{k}   + \pce{w}^j_{k}, \quad   \pce{z}^j_{0} = \pce{z}^j_{\ini}, \quad j\in \I_{[0,L-1]}.
\end{equation}

\subsection{Data-driven system representation via the stochastic fundamental lemma}
It is the linearity of the series expansion which ensures that the original stochastic system \eqref{eq:ARX} and its PCE formulation \eqref{eq:ARX_PCE} are structurally similar. Subsequently, we recall the stochastic variant of Willems' fundamental lemma which exploits this structural similarity. 

\begin{definition}[Persistency of excitation \citep{Willems2005}] Let $T, t \in \mbb{N}^+$. A sequence of real-valued inputs $\trar{u}{T-1}$ is said to be persistently exciting of order $t$ if the Hankel matrix
	\begin{equation*}
		\hankel_t(\trar{u}{T-1}) \doteq \begin{bmatrix}
			u\inst 0   &\cdots& u\inst{T-t} \\
			\vdots & \ddots & \vdots \\
			u\inst{t-1}& \cdots  & u\inst{T-1} \\
		\end{bmatrix}
	\end{equation*}
	is of full row rank. 
\end{definition}

\begin{lemma}[Stochastic fundamental lemma~\citep{Pan21s,Faulwasser2022}]\label{lem:RVfundamental} Let Assumption~\ref{ass:minimal_state_system} hold, and consider the stochastic LTI system \eqref{eq:ARX} and its $\mcl L^2(\Omega, \mcl F_k,\mu;\R^{\dimv})$, 
	$\dimv \in\{ \dimu, \dimw, \dimy\}$ trajectories of 
	random variables and the corresponding realization trajectories from \eqref{eq:ARX_realization}, which are $\trar{(U,W,Y)}{T-1}$ and $\trar{(u,w,y)}{T-1}$, respectively.  
	Let  $\trar{(u,w)}{T-1}$ be persistently exciting of order $\dimx +t$. Then
	$ (\widetilde{\mbf{U}}, \widetilde{\mbf{W}},\widetilde{\mbf{Y}})_{[0,t-1]}$ is a trajectory of \eqref{eq:ARX} if and only if there exists $G \in \splx{T-t+1} $ such that 
	\begin{equation} \label{eq:RVfunda}
		\hankel_t(\trar{v}{T-1}) G=\widetilde{\mbf{V}}_{[0,t-1]}
	\end{equation} 
holds	for all $(\mbf{v}, \widetilde{\mbf{V}})\in \{(\mbf{u},\widetilde{\mbf{U}}), (\mbf{w}, \widetilde{\mbf{W}}),(\mbf{y},\widetilde{\mbf{Y}})\} $. 
\end{lemma}
\begin{corollary}[PCE coefficients via realizations~\citep{Pan21s,Faulwasser2022}]\label{cor:mixed_funda}~\\
	Let the conditions of Lemma~\ref{lem:RVfundamental} hold. Then, for $j \in \I_{[0,L-1]}$, $\trar{ ( \tilde{\pce u}, \tilde{\pce w},\tilde{\pce y})}{t-1}^j $ is a PCE coefficient trajectory of \eqref{eq:ARX_PCE}
	if and only if there exists a $\pcecoe{g}{j}\in \R^{T-t+1}$ such that 
	\begin{equation} \label{eq:mixed_funda}
		\hankel_t(\trar{v}{T-1}) \pcecoe{g}{j}= \trar{ \tilde{\pce v}}{t-1}^j
	\end{equation}
	holds for all $(\mbf{v}, \tilde{\pce{v}})\in \{ (\mbf{u},\tilde{\pce{u}}), (\mbf{w},\tilde{\pce{w}}),(\mbf{y},\tilde{\pce{y}})\} $.
\end{corollary}
As the page limit prohibits the detailed discussion of the above we refer to   \cite{Pan21s,Faulwasser2022} for  proofs and examples. 
However, the key insight underlying the above results is that the dynamics of random variables \eqref{eq:ARX}, its realization dynamics \eqref{eq:ARX_realization}, and the dynamics of PCE coefficients \eqref{eq:ARX_PCE} share the same system structure and matrices. Hence, their finite-length trajectories can be linked to the trajectories of realization dynamics \eqref{eq:ARX_realization} as shown in \eqref{eq:RVfunda} and \eqref{eq:mixed_funda}. Put differently, in \eqref{eq:RVfunda} and \eqref{eq:mixed_funda} \textit{Hankel matrices in measured realization data} are used to predict the future evolution of random variables.

\section{ Data-driven stochastic output-feedback  predictive control} \label{sec:MPC}
In this section, we extend the  data-driven stochastic output-feedback  predictive control \cite{Pan2022}  through  interpolation of the initial condition. We provide sufficient conditions of recursive feasibility, performance, and stability of the extended scheme.

\subsection{Data-driven stochastic OCP with interpolated initial conditions}\label{sec:OCP}
Consider the stochastic LTI system~\eqref{eq:ARX}, its realization dynamics~\eqref{eq:ARX_realization}, and its PCE coefficients dynamics~\eqref{eq:ARX_PCE} with respect to the basis $\{\phi^j\}_{j=0}^{L-1}$, cf. \eqref{eq:finite_basis}.
 Suppose that a realization trajectory $\trar{(u,w,y)}{T-1}$ of \eqref{eq:ARX_realization} is available with $\trar{(u,w)}{T-1}$ persistently exciting of order $\dimx +N+\Tini$.
 
In the following, consider $	\pcecoe{v}{[0,L-1]}\doteq [\pcecoe{v}{0\top},\pcecoe{v}{1\top},\dots,\pcecoe{v}{L-1\top}]^\top$ as the vectorization of PCE coefficients over PCE dimensions, and let $\pce{v}^j_{i|k}$ be the predicted value of $\pce{v}^j_{k+i}$ at time instant $k$, $\pce{v}\in\{\pce{u},\pce{y},\pce{z}\}$. 
 At time instant $k \in \N$, given the PCE coefficients of the disturbances, i.e.  $\pce{w}_{k+i}^{[0,L-1]}$, $i \in \I_{[0,N-1]}$, the realization of the current extended state $z_k$, and the predicted PCE coefficients from last step $\pcecoe{z}{[0,L-1],\star}_{1|k-1}$, we consider the following data-driven stochastic OCP 
 \begin{subequations}\label{eq:PCEOCP}
		\begin{gather}
V_N\left(z_k,\pcecoe{z}{[0,L-1],\star}_{1|k-1}\right)\doteq		\min_{
					\pcecoe{(u,y,z)}{\cdot}_{\cdot|k},\pcecoe{g}{\cdot},\mu}\,
		\sum_{j=0}^{L-1}\norm{\phi^j}^2  \left(	\sum_{i=0}^{N-1}  \Big(\| \pcecoe{y}{j}_{i|k}\|^2_Q  +\|\pcecoe{u}{j}_{i|k}\|^2_R\Big)
		+ \|\pce{z}^j_{N|k}\|_P^2 \right)   \label{eq:H_PCE_SOCP_obj} 	\\
		\text{subject to }		\quad \left[\begin{array}{l} 
				\hankel_{N+\Tini}(\trar{u}{T-1})\\
				\hankel_{N+\Tini}(\trar{y}{T-1})\\
				\hankel_{N}(\mbf{w}_{[\Tini,T-1]})\\
		\end{array} \right]    
			\pcecoe{g}{j}	= \left[\begin{array}{l}    
				\pce{u}^j_{[-\Tini,N-1]|k}\\
				\pce{y}^j_{[-\Tini,N-1]|k}\\     
				\pce{w}^j_{[k,k+N-1]}\\
			\end{array} \right]
			,  \quad \forall  j\in \mbb{I}_{[0,L-1]}\label{eq:H_PCE_SOCP_hankel}
			\\		
							\pcecoe{z}{[0,L-1]}_{0|k}	\in \mbb{Z}_{\ini}\left(\mu,z_k,\pcecoe{z}{[0,L-1],\star}_{1|k-1}\right) \label{eq:PCEOCP_Ini}\\
		 	\pcecoe{v}{0}_{i|k}  \pm \sigma(\varepsilon_v)\sqrt{\sum_{j=1}^{L-1} {(\pcecoe{v}{j}_{i|k}})^2\langle \phi^j\rangle^2} \in \mathbb V, \quad \pce{v}\in\{\pce{u},\pce{y}\},\quad \forall i \in \mbb{I}_{[0,N-1]} \label{eq:PCE_chance}\\
			 \pce{u}_{i|k}^{j'}= 0,\quad \forall j'\in \I_{[L_{\text{ini}}+i(L_w-1)+1,L-1]}, \quad \forall  i\in \mbb{I}_{[0,N-1]}  \label{eq:causality},\\
 \pcecoe{z}{0}_{N|k} \in \Zf, \quad  \sum_{j=1}^{L-1}\pcecoe{z}{j\top}_{N|k}\Gamma\pcecoe{z}{j}_{N|k}\norm{\phi^j}^2\leq \gamma, 	\label{eq:terminal} \\
 \pcecoe{z}{j}_{i|k}=\begin{bmatrix}
 	\pcecoe{u}{j}_{[i-\Tini,i-1]|k}\\
 	\pcecoe{y}{j}_{[i-\Tini,i-1]|k} 
 \end{bmatrix}, \quad i \in \I_{[0,N]}, \quad j \in \I_{[0,L-1]}. 	\label{eq:uyz}
		\end{gather} 
	\end{subequations}
Observe that the decision variables $\pce{z}$ are redundant. They are introduced for the sake of compact notation and they can easily be avoided via  \eqref{eq:uyz}. 

The objective function \eqref{eq:H_PCE_SOCP_obj} penalizes the predicted input and output PCE coefficients with $R=R^\top\succ 0 $ and $Q=Q^\top\succ 0$. Moreover,  $P=P^\top \succ 0$ characterizes the terminal cost with respect to the PCE coefficients $\pce{z}^j_{N|k}$. Formulated in terms of PCE coefficients, the objective function \eqref{eq:H_PCE_SOCP_obj} is equivalent to the expected value of its counterpart with random variables \citep{Pan21s}. The linear equalities \eqref{eq:H_PCE_SOCP_hankel}--\eqref{eq:PCEOCP_Ini} encode the dynamics of the PCE coefficients~\eqref{eq:ARX_PCE} in a non-parametric fashion and based on measured data, cf. Lemma~\ref{lem:RVfundamental} and Corollary~\ref{cor:mixed_funda}.

The initial condition of OCP~\eqref{eq:PCEOCP} is considered in \eqref{eq:PCEOCP_Ini}. Given $\pce{z}^{[0,L-1],\star}_{1|k-1}$ as the predicted value of $\pce{z}^{j}_k$ with respect to the optimal solution at time instant $k-1$, we specify the constraint set in  \eqref{eq:PCEOCP_Ini} as 
\begin{equation}\label{eq:ini}
	\mbb{Z}_{\ini} \doteq \left\{\pcecoe{z}{[0,L-1]}_{0|k} \middle|\begin{gathered}
	\pcecoe{z}{0}_{0|k}= \mu z_k + (1-\mu)\pcecoe{z}{0,\star}_{1|k-1},\quad 0\leq\mu \leq 1 \\
	\sum_{j=1}^{L_{\ini}-1}\pcecoe{z}{j}_{0|k}\pcecoe{z}{j\top}_{0|k} \norm{\phi^j}^2 =(1-\mu)^2\sum_{j=1}^{L-1}\pcecoe{z}{j,\star}_{1|k-1}\pcecoe{z}{j,\star\top}_{1|k-1} \norm{\phi^j}^2\\
		\pcecoe{z}{j}_{0|k}=0,\quad j =\I_{[L_{\ini},L-1]}.
	\end{gathered} 
\right 	\}
\end{equation}
Note that in \eqref{eq:ini} we consider the initial condition $Z_{0|k}=\sum_{j=0}^{L-1} \pce{z}^j_{0|k}\phi^j$  to be a random variable rather than the current realization $z_k$. Specifically, as $\mu \in [0,1]$  we enforce the expected value and the covariance of $Z_{0|k}$ to be a convex combination  of $z_k$ and $Z^\star_{1|k-1}=\sum_{j=0}^{L-1} \pce{z}^{j,\star}_{1|k-1}\phi^j$ which is the predicted valued of $Z_k$. In addition, we can design the distribution of $Z_{0|k}$ by the choice of $\{\phi_{\ini}^j\}_{j=0}^{L_{\ini}-1}$ which is equivalent to $\{\phi^j\}_{j=0}^{L_{\ini}-1}$ as shown in \eqref{eq:finite_basis}. For the sake of efficient computation, we consider  $\{\phi^j_{\ini} \}_{j=1}^{L_\ini-1}$ to be $\dimz$ dimensional i.i.d. Gaussian, i.e.
\begin{equation}\label{eq:Gaussian_ini}
\phi_{\ini}^0=1,\quad	L_{\ini}=\dimz+1, \quad \phi^j_{\ini} \sim \mcl{N}(0,1),\quad \forall j \in \I_{[1,\dimz]}.
\end{equation}
This way, we have $Z_{0|k}$ also to be  Gaussian since it is a linear combination of  $\{\phi_{\ini}^j\}_{j=0}^{L_{\ini}-1}$. Notice that this choice does not prevent the consideration of non-Gaussian disturbances since the basis $\varphi^{j}$ for disturbances in \eqref{eq:finite_basis} is chosen according to the underlying distribution. 

Moreover, with  $\{\phi_{\ini}^j\}_{j=0}^{L_{\ini}-1}$ satisfying \eqref{eq:Gaussian_ini}, we reformulate the quadratic and nonconvex equality constraints of \eqref{eq:ini} in a linear fashion.  Specifically, note that the  right-hand-side matrix of the quadratic equality constraint
\begin{align*}
	Q_{\text{rhs}} &\doteq	\sum_{j=1}^{\tilde{L}_{\ini}-1}\pce{z}^{j,\star}_{1|k-1}\pce{z}^{j,\star\top}_{1|k-1} \norm{\phi^j}^2 = Q_{\text{rhs}} ^\top \succeq 0,
\end{align*} 
is known prior to solving the OCP. Hence, we can compute its eigen-decomposition $Q_{\text{rhs}} = U_{\text{rhs}}$ $D_{\text{rhs}}U_{\text{rhs}}^\top,$ where $D_{\text{rhs}}$ is a diagonal matrix with non-negative elements.
Then, one positive semi-definite solution of the non-convex quadratic equality constraint reads
\begin{equation}\label{eq:lineaRefor}
[\pce{z}^{1}_{0|k},\pce{z}^{2}_{0|k},  \dots,\pce{z}^{\dimz}_{0|k} ] = (1-\mu) U_{\text{rhs}}D_{\text{rhs}}^{\frac{1}{2}}U_{\text{rhs}}^\top,
\end{equation}
since $\norm{\phi^j}^2=\norm{\phi_{\ini}^j}^2=1$ holds for $j\in \I_{[1,\dimz]}$ if  \eqref{eq:Gaussian_ini} is considered. Substituting the second constraint of \eqref{eq:ini} by \eqref{eq:lineaRefor}, we avoid the non-convexity at the cost of computing one small-scale eigen-decomposition prior to each optimization.

Furthermore,  \eqref{eq:PCE_chance} is a usually conservative approximation of chance constraints with $\sigma_v = \sqrt{(2-\varepsilon_v)/\varepsilon_v}$, $v \in \{u,y\}$ \citep{farina13probabilistic}. The required probabilities are indicated by $1-\varepsilon_u$ and $1-\varepsilon_y$. We impose a causality constraint in \eqref{eq:causality} by considering $U_{i|k}$ as an affine feedback of $Z_{i|k}$.  The terminal constraints are specified in \eqref{eq:terminal}. Specifically, considering $Z_{N|k}\doteq \sum_{j=0}^{L-1}\pce{z}^j_{N|k}\phi^j$, the terminal constraints \eqref{eq:terminal} require the expected value of $Z_{N|k}$ to be inside of the terminal region $\Zf$ and the trace of its covariance weighted by $\Gamma$ to be smaller than $\gamma \in R^+$. Precisely, we assume that these terminal ingredients satisfy \cite[Assumption~3]{Pan2022}. For more details on a data-driven design of the terminal ingredients, we refer to the detailed discussions by \cite{Pan2022}.

\subsection{Predictive control scheme and closed-loop properties}\label{sec:PredictiveControl}

With the interpolated initial conditions,  we extend the output-feedback stochastic data-driven predictive control scheme proposed by \cite{Pan2022} based on OCP~\eqref{eq:PCEOCP}.

The predictive control scheme consists of the off-line data collection phase and the on-line optimization phase. In the off-line phase, a random input and disturbance trajectory $(\mbf{u},\mbf{w})_{[0,T-1]}$ is generated to obtain $\mbf{y}_{[0,T-1]}$. Note that the disturbance trajectory $\mbf{w}_{[0,T-1]}$ can also be estimated from output data, cf. \citep{Pan21s,Wang2022} for details. Moreover, the recorded input, output, and disturbance trajectories are used to determine the terminal ingredients and to construct the Hankel matrices for OCP~\eqref{eq:PCEOCP}. 

 In the on-line optimization phase, we assume the OCP~\eqref{eq:PCEOCP} is feasible at time instant $k=0$ with $ \pcecoe{z}{0,\star}_{1|-1} =z_0, \pcecoe{z}{[1,L-1],\star}_{1|-1}=0$ in \eqref{eq:PCEOCP_Ini} such that only the measured initial condition $z_0$ is considered. Then at each time $k$, we solve OCP~\eqref{eq:PCEOCP} to obtain the PCE coefficients of the first optimal input and the interpolated initial condition, i.e. $\pce{u}^{j,\star}_{0|k}$ and $\pce{z}^{j,\star}_{0|k}$, respectively. Notice that $\pcecoe{u}{j,\star}_{0|k} =0 $ holds for $j \in \I_{[L_{\ini}, L-1]}$ due to the causality condition \eqref{eq:causality}. Specifically, we consider the feedback input $u^{\text{cl}}_k$ as a realization of random variable given by $\pce{u}^{j,\star}_{0|k}$,
\begin{equation}\label{eq:feedback}
	u^{\text{cl}}_k = \sum_{j=0}^{L_{\ini}-1} \pce{u}^{j,\star}_{0|k} \phi^j_{\ini}\relx,
\end{equation} 
where  $\{\phi^j_{\ini}\relx\}_{j=1}^{L_{\ini}-1}$ are obtained by considering  the the current measured  $z_k$ as a realization of random variable given by $\pce{z}^{j,\star}_{0|k}$, cf. \citep{Pan2022}. Hence, the optimal solution to OCP~\eqref{eq:PCEOCP} is indeed an affine feedback policy in the realizations $\phi^j_{\ini}\relx$.

Applying the feedback input $u_{k}^\text{cl}$ determined by \eqref{eq:feedback}, the closed-loop dynamics of \eqref{eq:ARX_realization} for a given initial condition $z_{\ini}^\text{cl}$ and disturbances $\{w_k\}_{k\in \N}$ read
\begin{subequations}\label{eq:ARX_cl}
\begin{equation}\label{eq:ARX_realization_cl}
	y_{k}^\text{cl} = \Phi z_{k}^\text{cl}  + D u_{k}^\text{cl}  + w_{k}, \quad z_{0}^\text{cl}= z_{\ini}^\text{cl}.
\end{equation}	
Moreover, we obtain the sequence $V_{N,k} \in \R$, $k\in\N$ corresponding to the optimal value function of \eqref{eq:PCEOCP} evaluated in the closed loop.
Accordingly, considering a probabilistic initial condition $Z_{\ini}^\text{cl}$ and probabilistic disturbance $W_k$, $k\in \N$, we obtain the closed-loop dynamics in random variables
\begin{equation}\label{eq:ARX_cl_RM}
	Y_{k}^\text{cl} = \Phi Z_{k}^\text{cl} + D U_{k}^\text{cl} +W_{k}, \quad Z_{0}^\text{cl}= Z_{\ini}^\text{cl},
\end{equation}
\end{subequations}
where conceptually  the realization of $U_{k}^\text{cl}$ is  $u_{k}^\text{cl}$. Similarly, we define the probabilistic optimal cost $\mcl{V}_{N,k}\in \spl $ as  $\mcl{V}_{N,k}\relx= V_{N,k}$
The following theorem summarizes the closed-loop properties of the proposed scheme.
\begin{theorem}[Recursive feasibility and stability]\label{thm:stability}
	Consider the closed-loop dynamics \eqref{eq:ARX_cl} resulting from  the proposed predictive  control algorithm based on OCP~\eqref{eq:PCEOCP}.
	Suppose that at time instant $k=0$, 
	OCP~\eqref{eq:PCEOCP} is feasible with the initial condition in \eqref{eq:PCEOCP_Ini} as $ \pcecoe{z}{0,\star}_{1|-1} =z_0^\text{cl}=z_{\ini}^\text{cl}, \pcecoe{z}{[1,L-1],\star}_{1|-1}=0$. Then, OCP~\eqref{eq:PCEOCP} is feasible at all time instants $k \in \N^+$ with the initial condition in \eqref{eq:PCEOCP_Ini} updated with the current measured initial condition $z_k^\text{cl}$ and the predicted PCE coefficients $ \pcecoe{z}{[0,L-1],\star}_{1|k-1}$ based on the optimal solution of the previous solution obtained at $k-1$.

Moreover, let $\alpha \doteq \trace\left(\Sigma_W(Q+\widetilde{E}^\top P \widetilde{E})\right) \in \R^+$, then the following statements hold:
\begin{itemize}
	\item[(i)] The optimal performance index of OCP~\eqref{eq:PCEOCP} at consecutive time instants satisfies
\begin{equation}\label{eq:cost_decay}
\mean\left[\mcl{V}_{N,k+1} - \mcl{V}_{N,k}\right] \leq  - \mean \big[ \|U^\text{cl}_{k}\|^2_{R}+ \|Y^\text{cl}_{k}\|^2_{Q} \big]+\alpha.
\end{equation}
\item[(ii)] In addition, 
\begin{equation}\label{eq:average_cost_bound}
 \lim_{k\rightarrow\infty} \frac{1}{k}\sum_{i=0}^k \mean \big[ \|U^\text{cl}_{i}\|^2_{R}+ \|Y^\text{cl}_{i}\|^2_{Q} \big] \leq \alpha,
\end{equation}
 i.e., the averaged asymptotic cost of the proposed algorithm is bounded from above by $\alpha$.
\end{itemize}
\end{theorem}
\begin{proof}(Sketch)
	In \cite{Pan2022}, we present sufficient conditions for recursive feasibility and stability of the aforementioned predictive scheme with binary selection of initial condition, i.e. $\mu \in \{0,1\}$. The proof relies on the fact that  with $\mu=1$ OCP~\eqref{eq:PCEOCP} is recursively feasible, cf. \cite[Proposition~1]{Pan2022}.  Hence, with the interpolation condition \eqref{eq:ini}, i.e. $\mu \in [0,1]$, the recursive feasibility naturally holds since $\mu=1$ is included. Moreover, due to the optimization over the initial condition, the optimal cost $V_{N,k}$ with $\mu \in [0,1]$ is bounded from above by the one with $\mu \in \{0,1\}$, which allows to infer the stability results \eqref{eq:cost_decay}--\eqref{eq:average_cost_bound} via the proof of  \cite[Theorem~1]{Pan2022}.
\end{proof}

\section{Numerical Example}\label{sec:examples}
We consider  the LTI aircraft model given by \cite{maciejowski02predictive} exactly discretized with sampling time $t_s=0.5~\text{s}$. 	The system matrices are	
\begin{align*}
	\Phi &= \scalemath{0.8}{\begin{bmatrix} -\phantom{0}0.019	&-\phantom{0}1.440	&-\phantom{0}0.201&	\phantom{-0}0.256&	\phantom{-0}0.050	&\phantom{-0}0.160&	-\phantom{0}0.256&	\phantom{-0}0.0860\\
			\phantom{-0}0.711&	-\phantom{0}1.800&	-\phantom{0}4.773&	\phantom{-0}3.6875&	\phantom{-0}0.650&	\phantom{-0}2.982&	-\phantom{0}2.688&	\phantom{-0}1.707\\
			\phantom{-0}1.444&	-\phantom{0}26.922	&-\phantom{0}15.746&	\phantom{-}12.898&	\phantom{-0}2.319	&\phantom{-}10.461	&-12.897	&\phantom{-0}5.171
	\end{bmatrix}},
\end{align*}
and	$D = 0_{3\times 1}$ with $\dimy=3$, $\dimu=1$, $T_\ini=2$, and thus $\dimz=8$. Its minimal state-space representation with $\dimx=4$ can be found in \citep{Pan21s}.
As the simulated plant, we consider the system dynamics~\eqref{eq:ARX}, where $W_k, k\in \N$ are i.i.d. uniform random variables with their support on $[-0.01,0.01]\times [-1,1] \times [-0.1,0.1]$.
Note that $Y^j$ denotes the $j$-th element of $Y$. We impose a chance constraint on $Y^1$, i.e. $\mathbb{P}[Y^1 \in \mbb{Y}^1] \geq 1- \varepsilon_y,$
where $\mbb Y^1 = [-1,1]$, $ \varepsilon_y = 0.1$, and $\sigma(\varepsilon_y)=\sqrt{(2-\varepsilon_y)/\varepsilon_y}=4.359$. The weighting matrices in the objective function are $Q=\text{diag}([1, 1, 1])$ and $R = 1$. 

We apply the proposed scheme with prediction horizon $N=10$. In the data collection phase we record input-output trajectories of $90$ steps to  construct the Hankel matrices and to determine  terminal ingredients  $P$, $\Gamma$, $\gamma$ and $\Zf$, cf. \citep{Pan2022}.
To obtain an exact PCE for each component of $W_k$, we employ Lergendre polynomials component-wisely such that $L_w=4$.  As shown in \eqref{eq:Gaussian_ini}, we choose the basis for initial condition accordingly with $L_\ini=1+\dimz=9$. Thus, from \eqref{eq:finite_basis}, the dimension of the overall PCE basis as $L=39$. 

With $50$ different sampled sequences of disturbance realizations, we show the corresponding closed-loop realization trajectories of the proposed scheme in Figure~1(a). It can be seen that the chance constraint for output $Y^1$ is satisfied with a high probability, while  $Y^2$ and $Y^3$ converged to a neighborhood of $0$. The (over-time) averaged cost trajectories are depicts in Figure~1(b) in a semi-logarithmic plot. We see the (over-sampling) averaged asymptotic value is bounded above by $\alpha= 295.21$ corresponding to the chosen settings, and thus it is in line with the insights of Theorem \ref{thm:stability}.
To illustrate the evolution of statistical distributions of the closed-loop trajectories \eqref{eq:ARX_cl_RM}, we sample a total of $1000$ sequences of disturbance realizations and initial conditions around $[0,-100,0]^\top$. Then, we compute the corresponding closed-loop responses. The time evolution of the (normalized) histograms of the output realizations $y^2$ at $k=0, 5,10, 15, 20$ is depicted in Figure~1(b), where the vertical axis refers to the (approximated) probability density of $Y^2$. As one can see, the proposed control scheme controls the system to a narrow distribution of~$Y^2$ centred at~$0$.

\begin{figure}[t!]
	\centering
	\begin{tabular}[b]{c}
		\includegraphics[width=.45\linewidth,trim={0 5mm 0 2mm},clip]{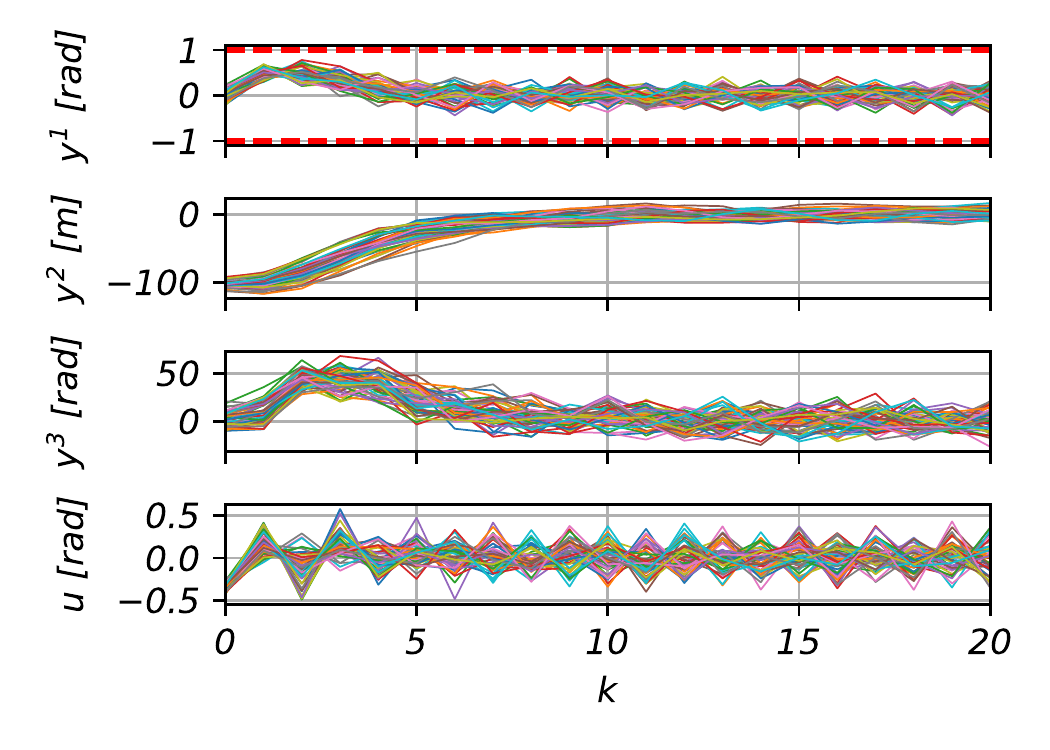} \\
		\small (a)\\
		\includegraphics[width=.45\linewidth,trim={0 5mm 0 2mm},clip]{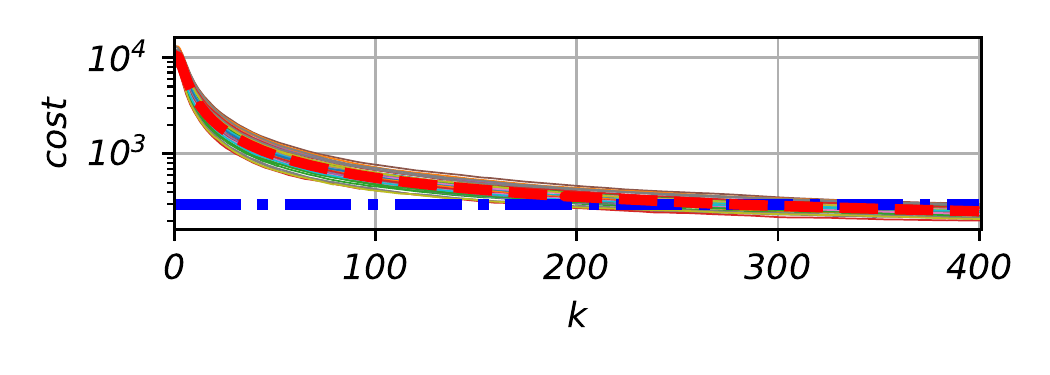} \\
		\small (b)
	\end{tabular} 
	\begin{tabular}[b]{c}
		\includegraphics[width=.45\linewidth]{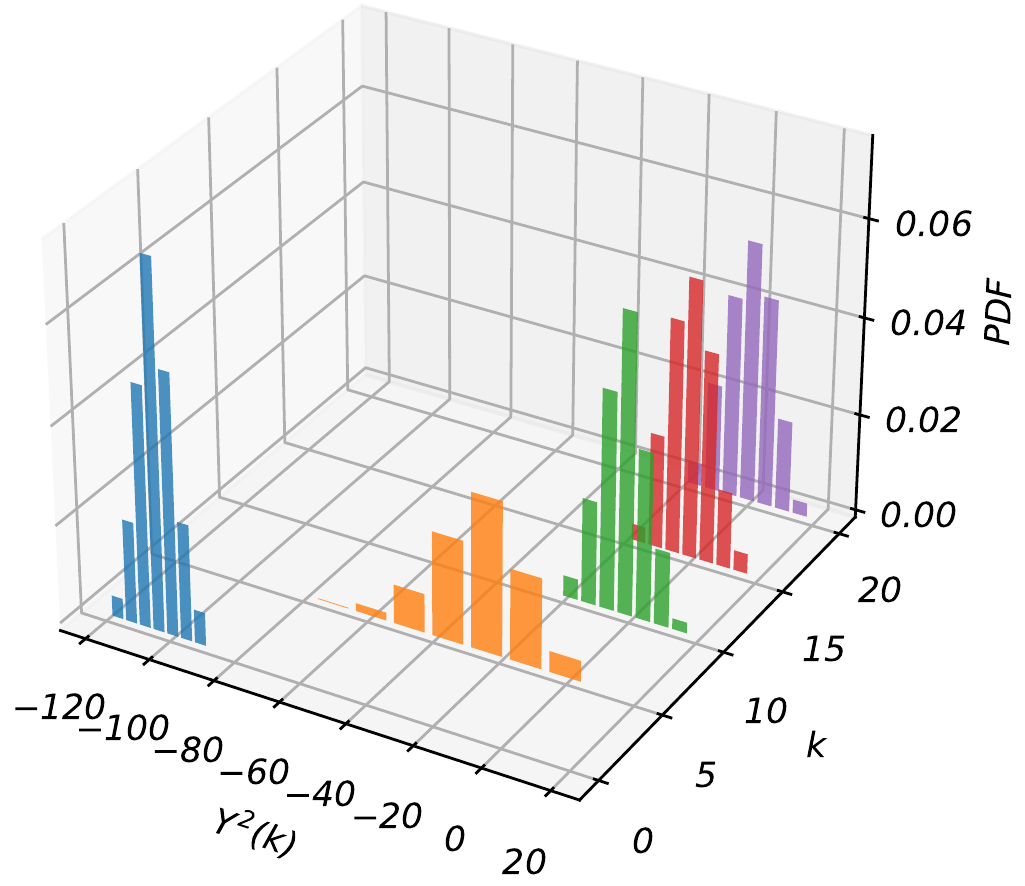} \\
		\small (c)
	\end{tabular}
	\caption{(a)~50 different closed-loop realization trajectories. The red-dashed lines represent the chance constraints.  (b)~Averaged cost (over time) of 50 different closed-loop realization trajectories. The red dashed lines represent the averaged asymptotic cost over sampling; the blue dash-dotted line represents the average asymptotic cost given by \eqref{eq:average_cost_bound}.
	(c)~Histograms of the output $Y^2$ from 1000 closed-loop realization trajectories.}
\end{figure}
\section{Conclusion}\label{sec:conclusion}
This paper has investigated data-driven stochastic output-feedback predictive control of linear time-invariant systems. We have shown that the concept of interpolating initial conditions can and should be exploited in the data-driven stochastic setting. Specifically, we have given sufficient conditions for recursive feasibility and practical stability as well as a corresponding performance bound. Numerical results illustrate the efficacy of the proposed approach. 
Our results, which are based on a recently proposed stochastic extension to Willems' fundamental lemma, underpin that stochastic predictive control can be formulated in data-driven fashion. 
Future work will consider tailored numerical methods for real-time feasible implementation and less restrictive stability conditions.

\acks{The authors gratefully acknowledge funding by the  German Research Foundation (Deutsche Forschungsgemeinschaft DFG) under project number 499435839.}

\bibliography{IEEEabrv.bib}

\end{document}